\documentclass[journal]{IEEEtran}
\usepackage{xcolor}
\usepackage{graphicx}
\usepackage{amsmath,amssymb,amsthm}
\usepackage[inline]{enumitem}
\usepackage{verbatim}

\newenvironment{myitemize}{\begin{list}{$\bullet$}
{\setlength{\topsep}{1mm}
\setlength{\itemsep}{0.25mm}
\setlength{\parsep}{0.25mm}
\setlength{\itemindent}{0mm}
\setlength{\partopsep}{0mm}
\setlength{\labelwidth}{15mm}
\setlength{\leftmargin}{4mm}}}{\end{list}}

\newtheorem{theorem}{Theorem}

\ifCLASSINFOpdf
\else
\fi
\begin{document}
%
\title{Cross-Layer Design of Automotive Systems}

\author{\IEEEauthorblockN{Zhilu Wang, Hengyi Liang, Chao Huang and Qi Zhu}\\
	\IEEEauthorblockA{Department of Electrical and Computer Engineering, Northwestern University}
}

%



\maketitle

\begin{abstract}
With growing system complexity and closer cyber-physical interaction, there are increasingly stronger dependencies between different function and architecture layers in automotive systems. This paper first introduces several cross-layer approaches we developed in the past for holistically addressing multiple system layers in the design of individual vehicles and of connected vehicle applications; and then presents a new methodology based on the weakly-hard paradigm for leveraging the scheduling flexibility in architecture layer to improve the system performance at function layer. The results of these works demonstrate the importance and effectiveness of cross-layer design for automotive systems.  
\end{abstract}


%
\IEEEpeerreviewmaketitle

\section{Introduction}
%
%
%
%

With the rapid development of active safety and autonomous functions, modern automotive systems have become complex cyber-physical systems that involve close interactions between the cyber domain (i.e., automotive electronic systems) and the physical domain (i.e., mechanical components and surrounding physical environment). The design and validation of these systems span across multiple layers, as illustrated in Fig.~\ref{fig:automotivelayers}. The function layer defines various automotive system functionalities in sensing, control, computation, communication, etc., and captures their interaction with the physical environment. In this paper, as we consider connected vehicle applications, the function layer can be further divided into the vehicular network layer and the individual vehicle function layer. The architecture layer defines the platform on which the system functionalities are implemented. It could include multiple sub-layers such as the software layer and the hardware layer. It may also include the layers of mechanical and physical components (e.g., engines, brakes, wires), but those are beyond the scope of this paper. In the AUTOSAR (Automotive Open System Architecture) standard, the automotive software layer could be further divided into a layer of runnables and a layer of software tasks connected with signals, as shown in~\cite{2015_ICCPS_Deng}. 
\begin{figure}[htbp]
	\centering
	\includegraphics[width=1\linewidth]{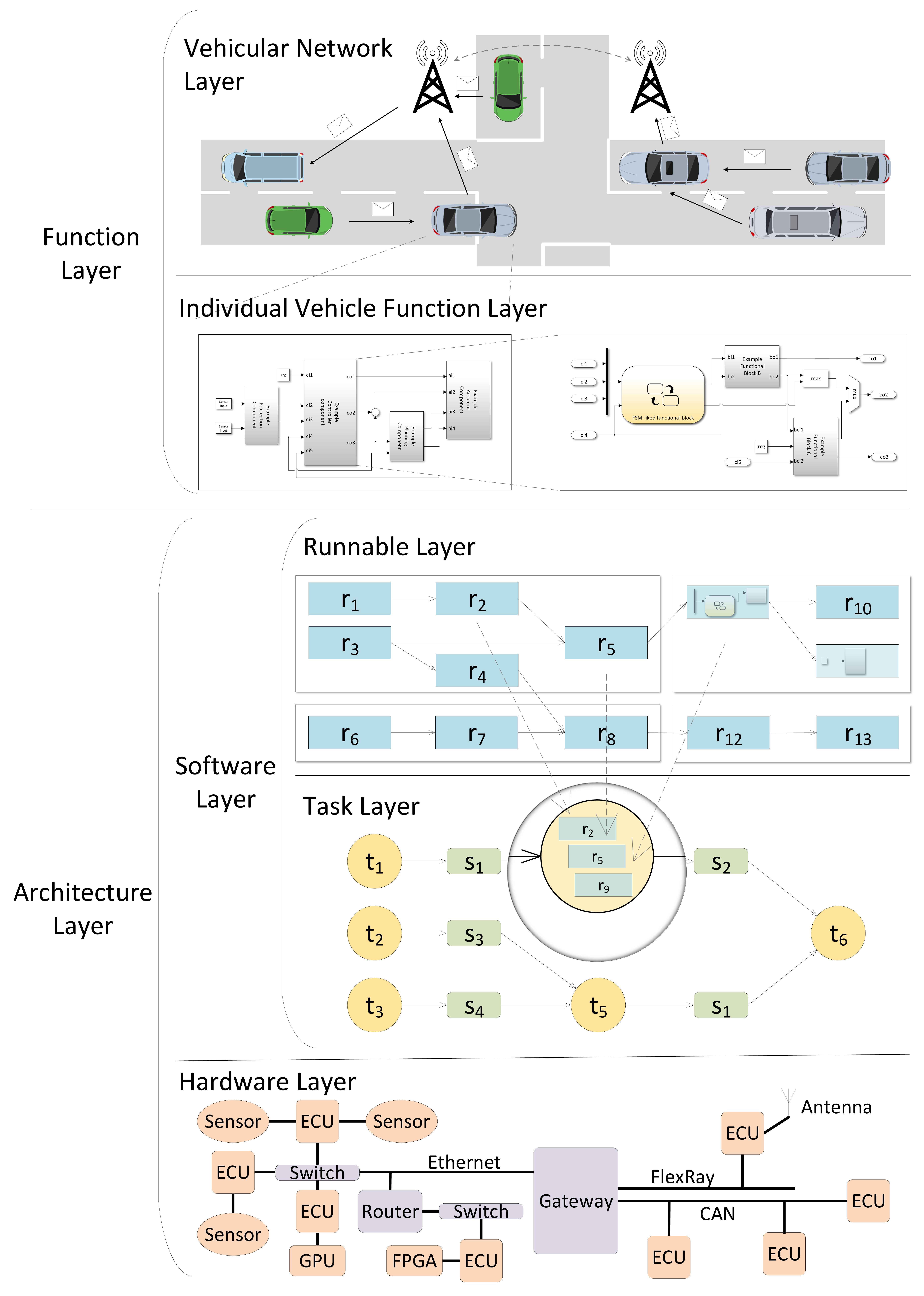}
	\caption{An illustration of the different layers for automotive system design, verification and validation. }
	\label{fig:automotivelayers}
\end{figure}

Traditionally, the design of different automotive layers is often carried out in an isolated fashion. However, the fast growing complexity of system functionality and architecture, as well as the close interaction between the cyber and physical domains, has led to strong dependency between layers and made those isolated approaches ineffective. For instance, whether an advanced control function or a new security feature can be deployed in a vehicle often depends on the availability of computation and communication resource at the architecture layer; whether a connected vehicle application can meet the safety and performance requirements depends on the communication delay and reliability of vehicular ad-hoc network. Thus, it is critical to adopt a cross-layer design methodology to holistically address the multiple layers in automotive systems.

In this paper, we will first introduce our previous works in cross-layer design for individual vehicles (Section~\ref{sec:individual_vehicles}) and for connected vehicles (Section~\ref{sec:connected_vehicles}). We will then present our initial results in developing a new cross-layer methodology for systems that allow the weakly-hard constraints (Section~\ref{sec:ongoing_works}).

\section{Cross-Layer Design for Individual Vehicles}
\label{sec:individual_vehicles}

Vehicle design spans across multiple layers. As shown in Fig.~\ref{fig:automotivelayers}, various automotive functionalities (e.g., sensing and control algorithms) can be captured at the function layer with formal or semi-formal models. These models are implemented at the architecture layer, often as software tasks running on the hardware platform. Traditionally, the functionality design, the generation of tasks from function models (which may go through a runnable layer as in AUTOSAR), and the mapping of tasks onto the hardware platform, are often done in isolated steps. However, as automotive systems are time-critical and resource-limited, design choices at the higher layers (e.g., functionality design or runnable generation) have significant impact on whether efficient or even feasible designs can be found at the lower layers. This has motivated our work in cross-layer design for individual vehicles, as introduced below. 

\subsection{Holistic Software Synthesis from Function to Architecture}

In~\cite{2015_ICCPS_Deng}, we propose a model-based software synthesis flow for AUTOSAR-based automotive systems. The cross-layer flow conducts runnable generation from the function model, task generation from the runnables, and task mapping onto a multicore platform in a holistic framework. Different from traditional approaches where runnable generation is performed merely from functional perspective and isolated from task generation and mapping, our approach explicitly addresses architectural properties in runnable generation, in particular regarding timing and schedulability.

In runnable generation, the functional blocks are mapped to runnables, as shown in Fig.~\ref{fig:automotivelayers}. Two algorithms are proposed in~\cite{2015_ICCPS_Deng} to explore different runnable generation options, while considering system modularity\footnote{As in~\cite{2015_ICCPS_Deng}, modularity is a metric that reflects the IP disclosure degree, and is measured by the number of runnables generated. A runnable generation can effectively hide the information of the internal block structure if the number of runnables (and their dependencies) is significantly smaller than the number of internal blocks. The optimal granularity is achieved when there is fewest number of generated runnables (when under certain constraints such as reusability and/or schedulability).}, reusability, code size, and schedulability: In a top-down method, a mixed integer linear programming (MILP) formulation is used to create the initial solution with the maximum modularity and reusability (i.e., no false input-output dependencies), and then a heuristic is used to decompose the runnables to improve schedulability. In a bottom-up method, another MILP formulation generates the initial solution with the maximum schedulability and a heuristic gradually merges runnables to optimize modularity.
To facilitate schedulability analysis, a formalism called Firing Execution Time Automaton (FETA) is developed, which can accurately capture the worst-case runnable timing behavior. In task generation and mapping, two algorithms are developed to group runnables into software tasks and map tasks onto Electronic Control Units (ECUs) on the hardware platform, while considering schedulability and memory cost for inter-task communication. For schedulability analysis, FETA is also applied at the task level to capture task timing behavior.

The experimental results in~\cite{2015_ICCPS_Deng} demonstrate that it is important to address timing and schedulability during the generation of runnables from function models, as the decisions at this stage already have significant impact on the eventual system feasibility. In particular, it is shown that there are strong trade-offs between modularity and schedulability. Previous methods that do not consider schedulability often lead to runnable generation solutions that have optimal modularity but are infeasible for task generation and mapping. Moreover, the proposed FETA formalism provides a holistic timing representation for functional blocks, runnables, and tasks, and is shown to be effective for schedulability analysis across these different layers.   

\subsection{Cross-Layer Design for Automotive Security}

Security has become a pressing issue for automotive systems in recent years, especially with the increase of vehicle connectivity and complexity. Various security protection mechanisms, such as message authentication, encryption, and anomaly detection, have been proposed for automotive systems. However, the successful deployment of these techniques depends on the available resources and whether the additional overhead may lead to timing violations of existing functions. 
It is thus important to take a cross-layer approach to address the design of security features together with other system objectives, including architecture layer properties such as timing and schedulability.

\subsubsection{Security-Aware Software Synthesis}

Traditional automotive software synthesis flow does not address security. 
It is often difficult to add security mechanisms after the software synthesis process is completed (i.e., after task allocation and scheduling are decided), because of the tight timing and resource constraints.
On the other hand, fixing the design of security mechanisms before software synthesis could often result in infeasible systems. Thus, we propose to address the security (function layer) together with the software synthesis (architecture layer) in an integrated formulation. 
In~\cite{2013_ICCAD_Lin, 2014_ICCAD_Lin}, we explore security mechanisms to protect communication messages against replay and masquerade attacks, which could happen when a malicious attacker compromises an ECU and then either replays legitimate messages on the communication bus or sends messages pretending as another ECU. Adopting message authentication codes (MACs) may prevent such attacks by authenticating each message with a key that only the message sender and receiver have, however they also incur overhead and could lead to timing and resource violations.

In~\cite{2013_ICCAD_Lin}, we consider adding MACs to Controller Area Networks (CAN) bus messages. Longer MACs make it harder for brute-force attacks, but also increase CAN message sizes and could lead to infeasible designs. 
To address these trade-offs, we develop an MILP-based algorithm to quantitatively explore the security design, including the messages to authenticate, MAC lengths and sharing strategies, together with the software synthesis options. Experimental results demonstrate that such holistic consideration can significantly increase the chance to find designs that satisfy both security and schedulability constraints -- although in some cases feasible solutions still cannot be found, given the limited bandwidth and message size of CAN.

In~\cite{2014_ICCAD_Lin}, we explore adding MACs to future automotive bus protocols that are based on TDMA communication (e.g., FlexRay, Time-Triggered Protocol, Time-Triggered Ethernet), which provide much higher bandwidths and larger message sizes. However, applying security mechanisms still requires careful analysis and design to avoid violations on other design constraints. In this work, we leverage the time-division property and adopt a key sharing mechanism that is based on time-delayed release of keys. 
This mechanism protects against masquerade attacks, however may lead to long message transmission delays. We quantitatively model the worst-case transmission delays under time-delayed release of keys, and develop a simulated annealing based method to holistically optimize task allocation and scheduling, TDMA-based network scheduling, and the key-release interval to meet both timing and security constraints.

\subsubsection{Security-Performance Trade-offs under Architectural Constraints}

With limited resources, improving automotive system security may require sacrificing other objectives, and such trade-off should be addressed in a quantitative and holistic manner. In~\cite{2016_TCAD_Zheng}, we explore the trade-off between security and control performance for a CAN-based system, while meeting schedulability constraints.

We consider a system model where multiple control tasks share the same ECU and communicate with sensors and actuators. Malicious attackers may eavesdrop on the messages between sensors and control tasks, and try to reconstruct the system state. This not only results in a loss of privacy, but could further be used as the basis for other attacks. Applying encryption techniques may prevent such attacks, however also introduces computation and communication overhead. For each encrypted message, a decryption task needs to be added on the ECU, which may force the control tasks to increase their activation periods to maintain schedulable. 

In~\cite{2016_TCAD_Zheng}, we present a cross-layer formulation to address system security level, control performance, and schedulability. The security level is the difficulty for attackers to restore the system states, measured through either Observability Gramian or analysis based on Kalman filter. We quantitatively model how the security level depends on the number of encrypted sensing channels (messages). On the other hand, for each control task, we model the relation between its performance and its period as an exponentially decaying function. Intuitively, having more encrypted messages increases system security level, but may lead to the increase of control task periods and thus worse control performance. 
We then develop a simulated annealing based algorithm to explore the choices of message encryption and control periods, under schedulability constraints. The experimental results demonstrate the clear trade-off between security and control performance, and show the importance of holistically considering these cross-layer properties.      

\section{Cross-Layer Design for Connected Vehicles}
\label{sec:connected_vehicles}

\begin{figure*}[htbp]
    \centering
    \includegraphics[width=\textwidth]{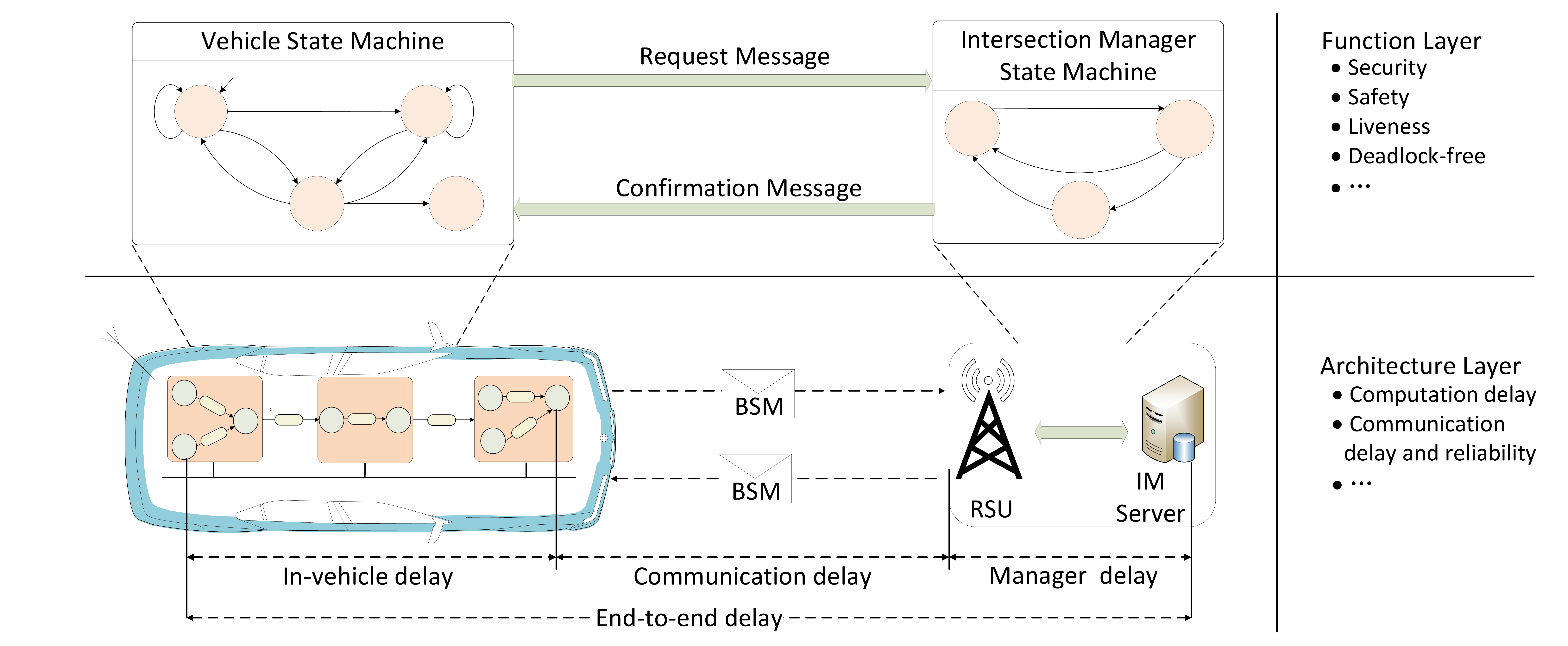}
    \caption{Cross-layer design for autonomous intersections.}
    \label{fig:cl_cv}
\end{figure*}

For many of the connected vehicle applications, the dependencies between different system layers need to be carefully considered during the design stage. In these applications, individual vehicles sense the external environment with on-board sensors, exchange information with nearby vehicles and infrastructures via wireless channels, analyze large amounts of data, and then conduct planning and control at real-time. The behavior of these applications at the function layer highly depends on the architecture layer properties such as the timing delay and reliability of wireless vehicle-to-vehicle (V2V) and vehicle-to-infrastructure (V2I) communications. 

To ensure the overall system correctness and performance, it is important to address function and architecture design in a holistic framework. 
Current practice, however, often addresses the two in isolation and overlooks the close dependency between function and architecture. To address this issue, we propose a framework in~\cite{2015_DAC_Zheng} to collaboratively conduct the verification of functionality and the synthesis of architecture platform. In particular, the framework includes two major aspects: \begin{enumerate*}[label={\alph*)},font={\color{red!50!black}\bfseries}]
\item function verification is carried out based on the assumptions of architecture layer properties (e.g., computation and communication delays are bounded within certain range), and 
\item architecture layer synthesis is performed while ensuring the assumptions specified in the function verification process are satisfied as constraints.
\end{enumerate*}
The interface between verification and synthesis is defined as a contract to formally capture the platform assumptions/constraints and other invariants. By systematically exploring the constraint settings in the interface and using them to drive function verification and architecture synthesis, our method can effectively reduce design and verification complexity, and identify designs that meet both functional and architectural requirements. 

In~\cite{2015_DAC_Zheng}, we perform a case study of collaborative adaptive cruise control (CACC), where leading and following vehicles exchange information such as velocity and acceleration through V2V messages to maintain a safe distance. At the function layer, a safety requirement is defined to ensure that the following vehicle can stop in time to avoid collision even if the leading vehicle brakes unexpectedly with its maximum braking power. 
Whether this safety requirement can be satisfied significantly depends on the end-to-end communication delay between the vehicles at the architecture layer.
Using our collaborative verification and synthesis approach, we can derive the constraint on inter-vehicle distance and end-to-end communication delay, which may be further refined to timing constraints on lower-layer architectural properties such as task activation periods and message transmission delays. 

We then apply the cross-layer methodology to the design and verification of centralized autonomous intersections in~\cite{bowen2017delay}. 
In this application, vehicles approaching an intersection exchange information with an intersection manager via V2I messages, and the intersection manager decides the order for vehicles to pass the intersection. It is essential to ensure that the system satisfies the following requirements at function layer while considering the V2I communication delay and message losses at architecture layer:
\begin{myitemize}
    \item \textit{Safety}: Any two vehicles with conflicting routes should not enter the intersection at the same time.
    \item \textit{Liveness}: A vehicle will eventually pass the intersection.
    \item \textit{Deadlock-free}: The intersection should not incur a deadlock where no vehicle can proceed.
\end{myitemize}
To guarantee these properties, we develop a delay-aware intersection management protocol, where end-to-end communication delay between vehicles and the intersection manager are explicitly modeled and the protocol is captured as timed automata for vehicles and the manager (as shown in Fig.~\ref{fig:cl_cv}). The protocol is proved to 1) always satisfy the safety requirement regardless of the delays, and 2) satisfy the liveness and deadline-free requirements if the end-to-end delay is bounded and the bound is known. 
Furthermore, through traffic simulations, it is shown that our autonomous intersection design can significantly outperform the tradition intersections~\cite{bowen2017delay}, as long as the maximum end-to-end delay is within a reasonable range (around 500ms in our study).
We believe that such quantitative results can help designers to set proper constraints for lower-layer architecture design.

\section{Cross-Layer Design with Weakly-Hard Paradigm}
\label{sec:ongoing_works}

Traditionally, the timing behavior of automotive functions has been specified based on hard constraints, where every instance of a task (or message) has to complete its execution (or transmission) by a pre-defined deadline. This is also the assumption in our previous works introduced in Section~\ref{sec:individual_vehicles} and~\ref{sec:connected_vehicles}. While such timing model facilitates worst-case analysis of system behavior, it is often over-pessimistic and rigid, resulting in infeasible or over-conservative designs. 

Many practical functions can in fact tolerate certain degrees of deadline misses, and their timing behavior can be described with the so-called \emph{weakly-hard} constraints, where bounded deadline misses are allowed. A common example is the $(m,\text{K})$ constraints, which specify that among any K consecutive instances of a task, at most $m$ of them can violate their execution deadlines~\cite{Bernat_TC_01}. Leveraging such weakly-hard constraints could more accurately define system timing requirements, significantly increase feasible design space under the typically-tight resource constraints in automotive systems, and improve design flexibility with additional timing slacks.

We believe that cross-layer design is particularly important for weakly-hard systems. To properly set the weakly-hard constraints (e.g., choose the values of $m$ and K) and effectively leverage their potential, it is essential to address the following two issues in a holistic manner: 1) at the function layer, ensure that system safety, stability, security, and other functional requirements can still be met under deadline misses allowed by weakly-hard constraints; 2) at the architecture layer, explore the design space under weakly-hard constraints to satisfy various system adaptation and retrofitting goals.  

In our recent work~\cite{huang2019formal}, we consider the first issue, and develop an approach for analyzing system functional properties under given degree of deadline misses. This work is different from works that focus on feedback controller synthesis for stability, such as~\cite{linsenmayer2017stabilization}. 
More specifically, our approach can determine whether a system is safe from an initial state under given $(m,\text{K})$ weakly-hard constraints. Previous verification methods could not be directly applied due to the lack of mechanism to capture and model the $(m,\text{K})$ specification at architecture level. To address this problem, our approach first carries out a series of transformations to abstract the $(m,\text{K})$ constraints, and then uses over-approximation based techniques to verify the safety of a new system model that combines the functional model and the abstraction. This approach is shown to be sound and effective in verifying system safety under weakly-hard constraints.
In another of our recent work~\cite{2019_ICCD_Liang}, we consider both issues, and develop a codesign approach to explore the addition of new security monitoring tasks by leveraging weakly-hard constraints for control tasks. The work studies the trade-off between control performance and system security level, when different degrees of deadline misses occur to the control tasks.

Next, we will introduce a novel cross-layer design approach for optimizing control sampling periods under weakly-hard constraints, to illustrate the potential of weakly-hard paradigm.

\medskip
\noindent
\textbf{Period Optimization with Weakly-hard Constraints:} There have been studies in the literature that explore control task periods across functional and architecture layers under the traditional hard timing constraints~\cite{cha2016control,Zhang2008RTSS}.
In this work, with the consideration of weakly-hard constraints, we can significantly expand the design space of control sampling periods for improving system feasibility and performance.

At function layer, reducing the sampling periods could typically lead to better control stability and performance, \emph{if} each control task instance can complete by its deadline. However, at architecture layer, shorter sampling periods also lead to higher resource utilization and may indeed cause deadline misses on some control tasks, which are detrimental to control stability and performance. It is thus important to study the cumulative effect of smaller periods and potential deadline misses under weakly-hard constraints (and vice versa) in exploring the design space. 

In below, we present a new cross-layer design approach for setting the sampling periods of control tasks while considering its impact at both function and architecture layers. We focus on theoretical control stability analysis at the function layer and task schedulability at the architecture layer. We also consider control performance via simulation in this work. The performance metric may be defined differently based on the control applications. In our example below, we measure the system control performance on its distance to the equilibrium. In work such as~\cite{2019_ICCD_Liang}, it is defined differently as the minimum time to reject a disturbance in the worst case (i.e., minimum time to bring the system back to equilibrium). 

\subsection{System model}
We consider a set of tasks $\{\tau_i\}$ running on a single ECU. All tasks are periodically activated. Each task $\tau_i$ is modeled by its period $T_i$, deadline $D_i$, and the execution time $C_i$. The system is scheduled by the static-priority preemptive policy. 

We consider a controller task $\tau_c$. The continuous-time dynamic of this linear time-invariant (LTI) system is:
\begin{equation}\label{equ:continuous}
    \dot{\textbf{x}}(t)=A_c\textbf{x}(t) + B_c\textbf{u}(t)
\end{equation}
where $\textbf{x}(t)\in\mathbb{R}^{n}$ and $\textbf{u}(t)\in\mathbb{R}^{m}$.

We assume that $\tau_c$ is running under the Logical Execution Time (LET) paradigm~\cite{2003_IEEE_Henzinger} where it receives the system state from sensors at the beginning of each sampling period and applies the control input to the actuators at the deadline. If a deadline miss occurs, the controller will apply the last calculated control input (from previous periods) at the deadline. For simplification, we assume that the deadline of this controller is the same as its period, i.e. $D_c=T_c$. And the discrete-time system dynamic is:
\begin{equation}
    \textbf{x}[k+1]=A\textbf{x}[k] + B\textbf{u}[k-p_k] \quad p_k=1,2,3,\ldots
\end{equation}
where
\begin{equation*}
    A=e^{A_c T_c}, \quad B = \int_0^{T_c} e^{A_c t} B_c dt
\end{equation*}
Here $\textbf{u}[k-p_k]$ is the latest control input at time $t=T_ck$, and $p_k$ is the related delay factor. For instance, $p_k=1$ if the deadline at $t=T_ck$ is not missed.

The control law is derived by solving the discrete-time linear–quadratic regulator (LQR) problem. Assume such control law is designed without considering any deadline misses.

The LET paradigm eases the control design as there will be a constant sensing-actuating delay.
By introducing the augmented state vector $\textbf{z}[k]=\left[\textbf{x}^\top[k], \textbf{u}^\top[k-1]\right]^\top$, the system dynamic used for solving the LQR is:
\begin{equation}
    \textbf{z}[k+1]=
    \begin{bmatrix}
     A & B\\
     \textbf{0} & \textbf{0}
    \end{bmatrix}
    \textbf{z}[k] + 
    \begin{bmatrix}
     \textbf{0}\\
     \textbf{I}
    \end{bmatrix}
    \textbf{u}[k]
    =A_z\textbf{z}[k] + B_z\textbf{u}[k]
\end{equation}
The control law $\textbf{u}[k]=-F\textbf{z}[k]$ is derived by minimizing the quadratic cost function:
\begin{equation}\label{equ:lqr_z}
    J=\sum_{k=1}^\infty \left(\textbf{z}^\top[k]Q\textbf{z}[k]+\textbf{u}^\top[k]R\textbf{u}[k] \right)
\end{equation}
where $Q$ and $R$ are both positive semi-definite matrices.

\subsection{Control stability under deadline misses}

Let $p_k$ denote the control input delay at time step $k$. Based on task schedulability analysis, we can deduce that $p_k$ is bounded by a maximum delay $\hat{p}=\left\lceil\frac{R_c}{T_c}\right\rceil$, where $R_c$ represents the task worst-case response time.
Then, we consider a new augmented state vector: $$\xi[k]=\left[\textbf{x}^\top[k],\textbf{u}^\top[k-1],\ldots,\textbf{u}^\top[k-\hat{p}]\right]^\top$$
we can have the system dynamic as:
\begin{equation}
    \xi[k+1] = A_\xi[k] \xi[k] + B_\xi u[k]
\end{equation}
\begin{equation}
    A_\xi[k]=
    \begin{bmatrix}
        A & B_1 & \ldots & B_{\hat{p}-1} & B_{\hat{p}}\\
        \textbf{0} & \textbf{0} & \ldots & \textbf{0} & \textbf{0}\\
        \textbf{0} & \textbf{I} & \ldots & \textbf{0} & \textbf{0}\\
         &  & \ddots &  & \\
         \textbf{0} & \textbf{0} & \ldots & \textbf{I} & \textbf{0}
    \end{bmatrix}
    ,\ 
    B_\xi=
    \begin{bmatrix}
        \textbf{0}\\
        \textbf{I}\\
        \textbf{0}\\
         \vdots \\
         \textbf{0}
    \end{bmatrix}
\end{equation}
where $B_{p_k}=B$ and $B_i=\textbf{0},\forall i\neq {p_k}$. As the control law $\textbf{u}[k]=-F\textbf{z}[k]$ is derived from the LQR problem~\eqref{equ:lqr_z}, we can rewrite the system dynamic as:
\begin{equation}\label{equ:wh_aug_dynamic}
    \xi[k+1] = (A_\xi[k] - B_\xi F_\xi)\xi[k]=\phi[k]\xi[k]
\end{equation}
where $F_\xi = \left[F,\textbf{0}\right],\ \textbf{0}\in \mathbb{R}^{(\hat{p}-1)m}$. 
As $B_\xi$ and $F_\xi$ are constant matrices and $A_\xi[k]$ is only related to $p_k$, there are $\hat{p}$ different $\phi[k]$, denoted as $\phi_1,\ldots,\phi_{\hat{p}}$. And we have $\phi[k]=\phi_{p_k}$.

The hyper-period $H_c$ of the controller task $\tau_c$ is the least common multiple of the periods of $\tau_c$ and its higher priority tasks. The execution pattern in each hyper-period is the same. As introduced in~\cite{2019_ICCD_Liang}, a weakly-hard schedulability analysis approach based on event-based simulation can derive the deadline miss pattern of the task in its hyper-period.
Moreover, during the simulation, the latest finished job can be recorded at each deadline, i.e. the delay factor $p_k$ will be recorded at $t=T_c*k$. Based on the consistency property of the hyper-period, we have $p_{k+iN_c} = p_k, \forall i\in \mathbb{Z}^+$, where $N_c = {H_c/T_c}$ is the number of jobs of $\tau_c$ in each hyper-period.

From the weakly-hard schedulability analysis, the delay factors $p_k$ in the hyper-period $k\in[0,N_c)$ are known. Thus, we have:
\begin{equation}
    \begin{aligned}
    \xi[k+N_c] &= \phi[k+N_c-1]\cdots\phi[k+1]\phi[k]\xi[k] \\
    &= \prod^0_{i=k+N_c-1}\phi[i] \prod^k_{j=N_c-1}\phi[j]\xi[k]\\
    &= \prod^0_{i=k+N_c-1}\phi_{p_i} \prod^k_{j=N_c-1}\phi_{p_j}\xi[k]\\
    &= \Phi_k\xi[k]
    \end{aligned}
\end{equation}

The following Theorem~\ref{thm:1} can then be used to check the control stability with deadline misses.

\begin{theorem}\label{thm:1}
The weakly-hard LTI system \eqref{equ:wh_aug_dynamic} is asymptotic stable if all eigenvalues of $\Phi_k$ are within the unit circle for all $k$:
\begin{equation}
    |\lambda^i_k| < 1,\ \forall \lambda^i_k\in eig(\Phi_k),\ \forall k\in[0,N_c)
\end{equation}
\end{theorem}

\begin{proof}
As all eigenvalues of $\Phi_k$ are in the unit circle, the sub-series $\xi'[l] = \xi[k+lN_c] = \Phi_k^l\xi[k],\ l=0,1,2,3,\dots$ is asymptotic stable, which can be expressed as:
\begin{equation} \label{equ:k_stable}
    \forall \varepsilon,\ \exists L_k(\varepsilon,\xi[k]),\ s.t.\  ||\xi[k+l N_c]|| \leq \varepsilon,\ \forall l \geq L_k
\end{equation}
As~\eqref{equ:k_stable} is satisfied by all $k\in [0,N_c)$, we have:
\begin{equation}
    \forall \varepsilon,\ \exists L(\varepsilon),\ s.t.\  ||\xi[k]|| \leq \varepsilon,\ \forall k \geq L
\end{equation}
where $L(\varepsilon) = \max\{k+N_c L_k(\varepsilon,\xi[k])|\forall k\in [0,N_c) \}$. Thus, the weakly-hard LTI system \eqref{equ:wh_aug_dynamic} is asymptotic stable.
\end{proof}

\subsection{Experiments}

In our experiments, we use a Furuta inverted pendulum as the example control plant to analyze the weakly-hard control functionality. The modeling of the Furuta pendulum is introduced in~\cite{2018_ECRTS_Pazzaglia}. 
The motor of the pendulum controls an arm that rotates in the horizontal plane. A pendulum is jointed to the arm and is free to rotate in the vertical plane. The system state is $\textbf{x}(t)=[\theta_r, \theta_p, \dot{\theta}_r, \dot{\theta}_p]^\top$, which are the angles of the arm and the pendulum, and their angular velocities. The control input $u(t)$ is the voltage applied to the motor. 
The continuous-time dynamic~\eqref{equ:continuous} of this Furuta pendulum in the numerical form is:
\begin{equation}
    A_c = 
    \begin{bmatrix}
     0 & 0 & 1 & 0 \\
     0 & 0 & 0 & 1 \\
     0 & 1.6907 & -2.9968 & -0.0048\\
     0 & 21.9176 & -3.0831 & 0.0626
    \end{bmatrix}
    , \ 
    B_c = \begin{bmatrix}
     0 \\
     0 \\
     3.8998 \\
     4.0122
    \end{bmatrix}
\end{equation}

Besides this pendulum control task, there are 7 regular periodic tasks that share the ECU. The periods of these tasks are varied from 60$ms$ to 300$ms$. 
The total utilization of these 7 regular tasks is $72\%$, where the utilization is $\sum C_i/T_i$.
The execution time of the controller task $\tau_c$ is $C_c=15\ ms$. We assume that all regular tasks do not allow any deadline misses, and they are scheduled with the rate-monotonic policy. 
The priority of $\tau_c$ is chosen to be the highest allowed one such that no low-priority regular tasks have deadline misses.

Under each period $T_c$, the control law $u[k]=-F\textbf{z}[k]$ is designed by solving the LQR problem \eqref{equ:lqr_z}, where $Q=diag(\alpha,\alpha,0,0,0)$ and $R = \beta$. $\alpha\in[0.1,10]$ and $\beta\in[0.1,1000]$ are the weights for states and control input in the quadratic cost. As the ratio $\alpha/\beta$ decreases, the control input intensity will reduce, while the convergence rate will be slower. We assume that the control input (i.e. the voltage apply to the motor) has an upper bound: $||F|| \leq 35$. During the control law design, it will find the highest $\alpha/\beta$ ratio that satisfies the $||F|| \leq 35$ constraint.

\begin{figure}[htbp]
	\centering
	\includegraphics[width=1\linewidth]{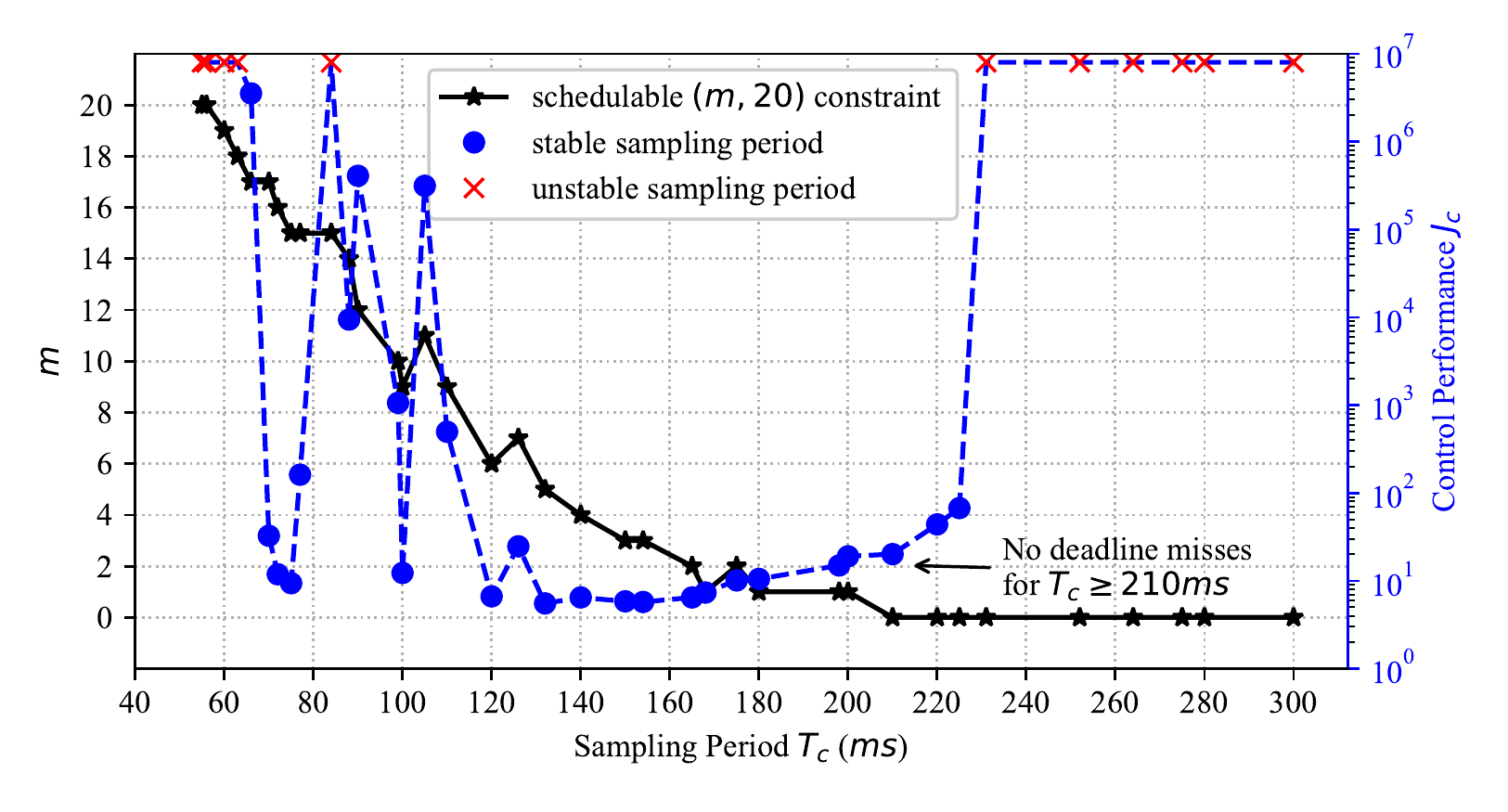}
	\caption{The black line (with star points) shows the smallest $m$ that makes the weakly-hard constraint $(m, \text{K})$ schedulable for $\text{K}=20$, under different sampling periods. The blue line (with round points and red crosses) shows the control stability and control performance $J_c$ for corresponding sampling period. The control input constraint is $||F|| \leq 35$. The results show that considering weakly-hard constraints can significantly expand the feasible design space and improve performance.}
	\label{fig:controlPerf_WHconst}
\end{figure}

We evaluate the system schedulability and control functionality for the sampling period $T_c$ from 55$ms$ to 300$ms$.
Besides control stability, we also evaluate the performance of the controller under different sampling period via simulation. Specifically, we generate 100 random initial states with random arrival time $t_0$ to represent disturbance and evaluate the cost among different sampling period. The cost $J_c$ is defined as the integral of the distance to the equilibrium:
\begin{equation}
    J_c = \int_{t_0}^{t_0 + \Delta t} \theta_r^2(t) + \theta_p^2(t) dt
\end{equation}
The control performance of each sampling period is the average cost among these $100$ cases.

Fig.~\ref{fig:controlPerf_WHconst} shows the smallest $m$ that makes the weakly-hard constraint $(m, \text{K})$ schedulable for each $T_c$ and the corresponding control stability and performance. We can see that the controller is stable for period $T_c\in[90,225]$. There is no deadline miss for period larger than $210ms$, i.e, the feasible design space under traditional hard deadlines is only $[210,225]$. With weakly-hard constraints, the space is expanded to at least $[90,225]$, with many period choices between $[60,90]$ feasible as well.
For sampling period from about $130ms$ to $170ms$, the control performance is close to the best performance, while the feasible weakly-hard constraint varies from $(5,20)$ to $(1,20)$. The best performance is achieved when $T_c$ is at $132ms$, with weakly-hard constraint $(5, 20)$. When the sampling period gets larger than $225ms$, the performance deteriorates and eventually control becomes unstable. When the sampling periods gets smaller than $120ms$, there are more deadline misses and the control becomes worse as well.

We also evaluate this system with different control input constraints, e.g., from $||F|| \leq 25$ to $\leq 50$. The results demonstrate that a lower input constraint will lead to shorter sampling period for better performance. For instance, if $||F|| \leq 30$, the system is unstable when $T_c\geq 200\ ms$ (i.e., when no deadline misses). It is stable and reaches the optimal performance when $T_c\in [120,140]$, even though some deadlines are missed.

This case study shows that leveraging weakly-hard constraints can expand the feasible design space of controller sampling periods (from [210,225] to at least [90, 225]) and achieve better performance (best performance achieved at $132ms$).
Moreover, to the focus of this paper, the results show that it is critical to address weakly-hard systems with a cross-layer approach that considers both function and architecture layers. It confirms a major motivation for using weakly-hard constraints, i.e., to leverage the robustness at function layer (e.g., the robustness of control functions with respect to occasional deadline misses) for expanding the design/adaptation flexibility at architecture layer (e.g., the flexibility to explore more sampling periods or add more security monitoring tasks). 

\section{Conclusion}

This paper presents several cross-layer methods for the design of automotive systems, including our prior works on systems with hard deadlines and our new results on weakly-hard systems. We believe that the strong dependencies between different function and architecture layers (even more so in weakly-hard systems) make it critical to take a cross-layer approach for addressing the design of automotive systems, and similar methodology might be applicable to other cyber-physical systems such as airplanes, robots, and various Internet-of-Things systems.


%



\section*{Acknowledgment}
The authors gratefully acknowledge the support from the National Science Foundation grants 1834701, 1834324, and 1839511, and the Office of Naval Research grant N00014-19-1-2496.


\bibliographystyle{IEEEtran}
\bibliography{bib/IEEEabrv,bib/crosslayer}
%



%


\begin{IEEEbiographynophoto}{Zhilu Wang}
received the B.S. degree in Applied Physics from University of Science and Technology of China, Hefei, China, in 2016. He is currently pursuing the Ph.D. degree at the Department of Electrical and Computer Engineering, Northwestern University, Evanston, IL, USA. His research interests include cyber-physical systems and software synthesis for real-time systems. 
He is a student member of the IEEE.
\end{IEEEbiographynophoto}


\begin{IEEEbiographynophoto}{Hengyi Liang}
received the B.S. degree in Microelectronics  from University of Electronic Science and Technology of China, Chengdu, China, in 2015. He is currently pursuing the Ph.D. degree with the Department of Electrical and Computer Engineering, Northwestern University, Evanston, IL, USA. His research interests include computer-aided design, software synthesis and intelligent transportation system. He is a student member of the IEEE.
\end{IEEEbiographynophoto}

\begin{IEEEbiographynophoto}{Chao Huang}
received a Ph.D. degree in Computer Science from Nanjing University, Nanjing, China, in 2018. He is currently a post-doctoral fellow with the Department of Electrical and Computer Engineering, Northwestern University, Evanston, IL, USA. His research interests include design and verification of model-based/data-driven autonomous systems.
\end{IEEEbiographynophoto}

\begin{IEEEbiographynophoto}{Qi Zhu} is an Associate Professor at the Department of Electrical and Computer Engineering, Northwestern University, Evanston, IL, USA. He has a Ph.D. in Electrical Engineering and Computer Science from the University of
California Berkeley, Berkeley, CA, USA. His research interests include design automation for intelligent cyber-physical systems (CPS) and Internet-of-Things (IoT) applications, cyber-physical security, machine learning for CPS/IoT, energy-efficient CPS, and system-on-chip design. He is a member of the IEEE.
\end{IEEEbiographynophoto}




\end{document}